\documentclass[12pt,a4paper]{article}

\usepackage[a4paper,margin=1in]{geometry}
\usepackage{amsmath,amssymb,amsthm}
\usepackage{setspace}
\usepackage[round,authoryear]{natbib}
\usepackage[colorlinks,citecolor=blue!60!black,linkcolor=blue!60!black,urlcolor=blue!60!black]{hyperref}
\usepackage{enumitem}
\usepackage{tikz}

\newtheorem{theorem}{Theorem}
\newtheorem{proposition}{Proposition}
\newtheorem{corollary}{Corollary}
\newtheorem{lemma}{Lemma}[section]
\theoremstyle{definition}
\newtheorem{definition}{Definition}
\newtheorem{example}{Example}
\newtheorem{remark}{Remark}

\newcommand{\pref}{\succsim}
\renewcommand{\L}{L}
\newcommand{\R}{R}
\newcommand{\qf}{q^{f}}
\newcommand{\qt}{q^{\tau}}
\newcommand{\xt}{\overline{x}^{\tau}}

\title{Envy-free, efficient and strategy-proof division\\ of a heterogeneous good between two agents}

\author{Mihir Bhattacharya\thanks{Department of Economics, Ashoka University, Rajiv Gandhi Education City, Rai, Sonipat, Haryana 131029, India. Email: \texttt{mihir.bhattacharya@ashoka.edu.in}}
\and
Ojasvi Khare\thanks{Department of Economics, Shiv Nadar Institution of Eminence, Gautam Buddha Nagar, Uttar Pradesh 201314, India. Email: \texttt{ojasvi.khare@snu.edu.in}.\newline
We thank Debasis Mishra for valuable feedback at various stages of the paper, and Bhaskar Dutta, Saptarshi Mukherjee, Ariel Rubinstein, Arunava Sen and Rohit Vaish for comments and suggestions. During the preparation of this work the authors used Google Gemini and Claude (Anthropic) for assistance with discussion, literature search, exposition and copy editing. After using these tools, the authors reviewed and edited the content as needed and take full responsibility for the content of the publication. The authors declare no conflict of interest. This research received no specific grant from any funding agency in the public, commercial, or not-for-profit sectors.}}

\date{}

\begin{document}
\maketitle

\begin{abstract}
\noindent
We study the division of $[0,1]$ between two agents, each receiving a closed subinterval, under ordinal preferences that need not be additive. A monotonic preference has a unique balanced cut, at which the agent is indifferent between the left and the right piece. On a domain of monotonic preferences, a strategy-proof, envy-free and Pareto efficient rule exists if and only if the balanced cuts of the preferences in the domain take at most two distinct values; when they take exactly two, the divide-and-choose rules are the only such rules. For preferences single-peaked in quantity, which are additive only if monotonic, the three properties are compatible when the agents' most preferred bundles can be assigned simultaneously.
\end{abstract}


\newpage

\section{Introduction}

We consider the division of a heterogeneous, perfectly divisible good, such as land or broadcasting time, represented by $[0,1]$, between two agents. Each agent receives a \emph{bundle}, a closed interval $[x,x+q]\subseteq[0,1]$ written $(x;q)$; the set of bundles is the triangle $T=\{(x;q)\in[0,1]^2:x+q\le1\}$ (Figure~\ref{fig:triangle}). Preferences are continuous weak orders on $T$ and need not be additive. We study rules that are strategy-proof, envy-free \citep{foley1967resource, varian1974equity} and Pareto efficient.

A preference is \emph{monotonic} if every bundle is strictly preferred to each of its proper subintervals. A monotonic preference has a unique \emph{balanced cut}, a point $\qf\in(0,1)$ such that the agent is indifferent between the left piece $[0,\qf]$ and the right piece $[\qf,1]$. Theorem~\ref{thm:sp} states that, on a domain of monotonic preferences, a strategy-proof, envy-free and Pareto efficient rule exists if and only if the balanced cuts of the preferences in the domain take at most two distinct values. Each preference has a single balanced cut, but a domain may contain many preferences with the same balanced cut; the condition restricts the number of distinct values, not the number of preferences. Necessity holds already for three additive preferences with distinct balanced cuts. If the balanced cuts take exactly two values, the only such rules are the two divide-and-choose rules, in which one agent divides the good at her balanced cut (Proposition~\ref{prop:unique}).

A preference is \emph{single-peaked in quantity} if, among bundles with left end-point $0$ and among bundles with right end-point $1$, it is single-peaked in length. Such a preference is additive only if it is monotonic (Proposition~\ref{prop:nonadditive}). On domains of these preferences, the three properties are incompatible whenever the domain contains three monotonic preferences with distinct balanced cuts (Corollary~\ref{cor:ssp}), and compatible whenever the agents' most preferred bundles can be assigned simultaneously (Proposition~\ref{prop:top}).

\paragraph{Related literature.}
Envy-free divisions of $[0,1]$ into intervals exist for any number of agents with continuous preferences that never favour an empty piece \citep{stromquist1980cut, su1999rental, segalhalevi2021groups}; when pieces need not be intervals, envy-free and Pareto efficient divisions exist for additive \citep{weller1985fair} and non-additive \citep{husseinov2013existence} utilities; see also \citet{berliant1992fair}. These non-additive results, and those of \citet{igarashi2024tangled}, assume monotonicity, which preferences single-peaked in quantity violate; \citet{barman2025beyond} study envy-free division under non-monotone valuations but not efficiency or incentives. Proposition~\ref{prop:fair} is elementary; it is stated because Theorem~\ref{thm:sp} uses the assignment of pieces it determines. For piecewise-uniform valuations, \citet{chen2013truth} construct a strategy-proof, envy-free and Pareto efficient mechanism and \citet{maya2012incentive} characterize the strategy-proof and efficient mechanisms for two agents; both allow pieces that are not intervals, and a monotonic piecewise-uniform valuation is uniform, with balanced cut $\tfrac12$. \citet{bu2023existence} show that no deterministic, truthful and proportional mechanism exists for two agents with hungry, piecewise-constant valuations; since an envy-free and efficient allocation is then proportional, Corollary~\ref{cor:sp} also follows from their theorem. Theorem~\ref{thm:sp} identifies the domains on which the three properties are incompatible. Single-peaked preferences over quantity are standard for a homogeneous good \citep{sprumont1991division, thomson1994resource}; in cake cutting, \citet{wang2019single} and \citet{bhardwaj2020fairness} study additive densities that are single-peaked in \emph{location}.

\section{The model}

A heterogeneous, perfectly divisible good is distributed on $[0,1]$. The set of agents is $N=\{1,2\}$. A \emph{bundle} is an interval $[x,x+q]\subseteq[0,1]$, denoted $(x;q)$, where $x$ is its location and $q$ its quantity. The set of bundles is $T=\{(x;q)\in[0,1]^2: x+q\le1\}$. An \emph{allocation} is a pair $y=(y_1,y_2)$ of bundles, where $y_i$ is agent $i$'s bundle. It is \emph{feasible} if the interiors of the two bundles are disjoint. The set of feasible allocations is $\mathcal{A}$. A feasible allocation need not exhaust $[0,1]$. For a cut $t\in[0,1]$ we write
\[
\L(t)=(0;t)=[0,t]\qquad\text{and}\qquad \R(t)=(t;1-t)=[t,1]
\]
for the left and right pieces.

\paragraph{Preferences.}
Let $\mathcal{C}$ be the set of complete, transitive and continuous preference relations $\pref$ on $T$; continuity means that upper and lower contour sets are closed. We use two domains.

\smallskip
\noindent\emph{Monotonic preferences.} The domain $\mathcal{D}\subseteq\mathcal{C}$ consists of the preferences satisfying:
\begin{enumerate}[label=A\arabic*., leftmargin=2.5em, nosep]
\item (Zero quantity is indifferent across locations.) $(x;0)\sim(0;0)$ for all $x\in[0,1]$.
\item (Set-inclusion monotonicity.) If $[x',x'+q']\subsetneq[x,x+q]$, then $(x;q)\succ(x';q')$.
\end{enumerate}
For example, the utility functions $u(x;q)=q$, $u'(x;q)=q(1-x)$ and $u''(x;q)=q(x+q)$ represent preferences in $\mathcal{D}$.

\smallskip
\noindent\emph{Preferences single-peaked in quantity.} The domain $\mathcal{S}\subseteq\mathcal{C}$ consists of the preferences satisfying A1 and:
\begin{enumerate}[label=A\arabic*$'$., leftmargin=2.5em, nosep, start=2]
\item (Positive quantity at the left end beats zero quantity.) $(0;q)\succ(0;0)$ for all $q\in(0,1]$.
\end{enumerate}
\begin{enumerate}[label=A\arabic*., leftmargin=2.5em, nosep, start=3]
\item (Single-peakedness in quantity.) There exist $\qt\in(0,1]$ and $\xt\in[0,1)$ such that the bundles $\L(\qt)$ and $\R(\xt)$ are most preferred in $T$, and
\begin{enumerate}[label=(\roman*), nosep]
\item $\qt\le q'<q''$ or $q''<q'\le\qt$ implies $\L(q')\succ\L(q'')$;
\item $\xt\le x'<x''$ or $x''<x'\le\xt$ implies $\R(x')\succ\R(x'')$.
\end{enumerate}
\end{enumerate}
The peaks $\qt$ and $1-\xt$ at the two ends may differ, and no condition is imposed on location. Every preference in $\mathcal{D}$ satisfies A2$'$ and A3 with $\qt=1$ and $\xt=0$; hence $\mathcal{D}\subseteq\mathcal{S}$.

\paragraph{Rules and axioms.}
For a domain $\mathcal{D}'\subseteq\mathcal{S}$, a \emph{rule} is a function $f:\mathcal{D}'^2\to\mathcal{A}$. At a profile ${\pref}=(\pref_1,\pref_2)$, an allocation $y\in\mathcal{A}$ is \emph{Pareto efficient} if there is no $z\in\mathcal{A}$ with $z_i\pref_i y_i$ for both $i$ and $z_j\succ_j y_j$ for some $j$; it is \emph{envy-free} if $y_i\pref_i y_j$ for both $i\neq j$. A rule is Pareto efficient (PE) or envy-free (EF) if $f(\pref)$ has the property at every profile. A rule is \emph{strategy-proof} (SP) if $f_i(\pref_i,\pref_{-i})\pref_i f_i(\pref_i',\pref_{-i})$ for all $i\in N$, all ${\pref}\in\mathcal{D}'^2$ and all $\pref_i'\in\mathcal{D}'$. Proofs not given in the text are in the Appendix.

\section{Fair allocations with monotonic preferences}

\begin{definition}[Balanced cut]
For $\pref\in\mathcal{D}$, a \emph{balanced cut} is a point $\qf\in[0,1]$ such that $\L(\qf)\sim\R(\qf)$.
\end{definition}

\begin{proposition}\label{prop:cut}
Every $\pref\in\mathcal{D}$ has a unique balanced cut $\qf\in(0,1)$. Moreover, $\R(t)\succ\L(t)$ for $t<\qf$ and $\L(t)\succ\R(t)$ for $t>\qf$.
\end{proposition}

We write $\qf(\pref)$ for the balanced cut of $\pref$ and $\qf_i=\qf(\pref_i)$ for agent $i$'s balanced cut. The \emph{balanced region} is the interval of cuts $[\min\{\qf_1,\qf_2\},\max\{\qf_1,\qf_2\}]\subset(0,1)$ (Figure~\ref{fig:triangle}).

\begin{figure}[t]
\centering
\begin{tikzpicture}[scale=5.2, >=stealth]
  \draw[->] (0,0) -- (1.12,0) node[right] {$x$};
  \draw[->] (0,0) -- (0,1.12) node[above] {$q$};
  \draw[thick] (0,1) -- (1,0);
  \node[left] at (0,1) {\small $(0;1)$};
  \node[below] at (1,0) {\small $(1;0)$};
  \node[below left] at (0,0) {\small $0$};
  \def\qa{0.38} \def\qb{0.58}
  \draw[blue!70!black, thick] (0,\qa) .. controls (0.12,0.30) and (0.28,0.45) .. (\qa,1-\qa);
  \draw[green!50!black, thick] (0,\qb) .. controls (0.20,0.40) and (0.45,0.36) .. (\qb,1-\qb);
  \draw[red!75!black, line width=2.2pt] (0,\qa) -- (0,\qb);
  \draw[red!75!black, line width=2.2pt] (\qa,1-\qa) -- (\qb,1-\qb);
  \fill (0,\qa) circle (0.25pt) node[left] {\small $(0;q^f_1)$};
  \fill (0,\qb) circle (0.25pt) node[left] {\small $(0;q^f_2)$};
  \fill (\qa,1-\qa) circle (0.25pt) node[above right] {\small $(q^f_1;1-q^f_1)$};
  \fill (\qb,1-\qb) circle (0.25pt) node[above right] {\small $(q^f_2;1-q^f_2)$};
  \draw[gray, dashed] (0.62,0) -- (0.62,0.2) -- (0.82,0);
  \fill[gray] (0.62,0.2) circle (0.22pt) node[above] {\small $(x;q)$};
\end{tikzpicture}
\caption{The set of bundles $T$. The point $(x;q)$ represents $[x,x+q]$; its subintervals form the dashed triangle. Left pieces $\L(t)$ lie on the line $x=0$ and right pieces $\R(t)$ on the line $x+q=1$. Each curve is the indifference curve of agent $i$ through $\L(q^f_i)$ and $\R(q^f_i)$. For $q^f_1<q^f_2$, the envy-free and Pareto efficient allocations are $(\L(\alpha),\R(\alpha))$ with $\alpha\in[q^f_1,q^f_2]$ (thick segments).}
\label{fig:triangle}
\end{figure}

\begin{proposition}\label{prop:fair}
Let ${\pref}\in\mathcal{D}^2$. An allocation is Pareto efficient and envy-free at $\pref$ if and only if one of the following holds:
\begin{enumerate}[label=(\roman*), nosep]
\item $\qf_1=\qf_2=\qf$ and the allocation is $(\L(\qf),\R(\qf))$ or $(\R(\qf),\L(\qf))$;
\item $\qf_1<\qf_2$ and the allocation is $(\L(\alpha),\R(\alpha))$ for some $\alpha\in[\qf_1,\qf_2]$;
\item $\qf_1>\qf_2$ and the allocation is $(\R(\alpha),\L(\alpha))$ for some $\alpha\in[\qf_2,\qf_1]$.
\end{enumerate}
\end{proposition}

If $\qf_1\neq\qf_2$, the agent with the smaller balanced cut receives the left piece in every envy-free and Pareto efficient allocation; if $\qf_1=\qf_2$, both assignments occur. In particular, the equal division $(\L(\tfrac12),\R(\tfrac12))$ is envy-free and Pareto efficient if and only if $\qf_1\le\tfrac12\le\qf_2$.

\section{Strategy-proofness}

For a domain $\mathcal{D}'\subseteq\mathcal{D}$, let $Q(\mathcal{D}')=\{\qf(\pref):{\pref}\in\mathcal{D}'\}$ be the set of balanced cuts of preferences in $\mathcal{D}'$. Distinct preferences in $\mathcal{D}'$ may have the same balanced cut, so $Q(\mathcal{D}')$ may be finite even if $\mathcal{D}'$ is infinite.

\begin{theorem}\label{thm:sp}
Let $\mathcal{D}'\subseteq\mathcal{D}$. There exists a strategy-proof, envy-free and Pareto efficient rule $f:\mathcal{D}'^2\to\mathcal{A}$ if and only if $Q(\mathcal{D}')$ has at most two elements.
\end{theorem}

\paragraph{Sufficiency.} If $Q(\mathcal{D}')=\{c\}$, the constant rule $f(\pref)=(\L(c),\R(c))$ is envy-free and Pareto efficient by Proposition~\ref{prop:fair}(i), and it is strategy-proof because it does not depend on the reported preferences. If $Q(\mathcal{D}')=\{l,h\}$ with $l<h$, define the \emph{divide-and-choose} rule by
\[
f(\pref)=\begin{cases}(\L(l),\R(l)) & \text{if } \qf_1=l,\\ (\R(h),\L(h)) & \text{if } \qf_1=h.\end{cases}
\]
Since $\qf_2\in\{l,h\}$, Proposition~\ref{prop:fair} implies that $f(\pref)$ is envy-free and Pareto efficient at every profile. The allocation does not depend on $\pref_2$, so agent~2 cannot manipulate. Agent~1's report determines only whether she receives $\L(l)$ or $\R(h)$. If $\qf_1=l$, then $\L(l)\sim_1\R(l)\succ_1\R(h)$; if $\qf_1=h$, then $\R(h)\sim_1\L(h)\succ_1\L(l)$. Hence $f$ is strategy-proof.

\paragraph{Necessity.} Suppose that $\mathcal{D}'$ contains preferences $P_l,P_m,P_h$ with balanced cuts $l<m<h$, and that $f$ is strategy-proof, envy-free and Pareto efficient. Proposition~\ref{prop:fair} and strategy-proofness determine $f$ along a sequence of profiles, ending at a profile at which an agent with balanced cut $l$ gains by reporting $P_m$ (Appendix).

\begin{corollary}\label{cor:sp}
There is no strategy-proof, envy-free and Pareto efficient rule on $\mathcal{D}^2$. The same holds on every domain of additive preferences that contains three preferences with distinct balanced cuts; for $c\in(0,1)$, the additive preference with density $\tfrac{1}{2c}$ on $[0,c]$ and $\tfrac{1}{2(1-c)}$ on $[c,1]$ has balanced cut $c$.
\end{corollary}

Each pair of the three properties is satisfied by some rule on $\mathcal{D}^2$ (Remark~\ref{rem:indep} in the Appendix).

\begin{remark}[Pareto efficiency and exhaustion of the good]\label{rem:nw}
On $\mathcal{D}^2$, an envy-free allocation $y$ is Pareto efficient if and only if $y_1\cup y_2=[0,1]$ (Lemma~\ref{lem:whole} and the proof of Proposition~\ref{prop:fair}), so Theorem~\ref{thm:sp} holds with Pareto efficiency replaced by this condition. On $\mathcal{S}$ the equivalence fails: the top rule of Proposition~\ref{prop:top} may leave part of $[0,1]$ unallocated.
\end{remark}

The following result identifies all rules with the three properties on the domains admitted by Theorem~\ref{thm:sp}. For $Q(\mathcal{D}')=\{l,h\}$ with $l<h$, let $f^1$ denote the divide-and-choose rule defined above, in which agent~1 divides the good at her balanced cut, and let $f^2$ denote the rule in which agent~2 does so:
\[
f^2(\pref)=\begin{cases}(\R(l),\L(l)) & \text{if } \qf_2=l,\\ (\L(h),\R(h)) & \text{if } \qf_2=h.\end{cases}
\]

\begin{proposition}\label{prop:unique}
Let $\mathcal{D}'\subseteq\mathcal{D}$.
\begin{enumerate}[label=(\roman*), nosep]
\item If $Q(\mathcal{D}')=\{c\}$, a rule $f$ on $\mathcal{D}'^2$ is strategy-proof, envy-free and Pareto efficient if and only if $f(\pref)\in\{(\L(c),\R(c)),(\R(c),\L(c))\}$ for every ${\pref}\in\mathcal{D}'^2$.
\item If $Q(\mathcal{D}')=\{l,h\}$ with $l<h$, a rule on $\mathcal{D}'^2$ is strategy-proof, envy-free and Pareto efficient if and only if it is $f^1$ or $f^2$.
\end{enumerate}
\end{proposition}

Together with Theorem~\ref{thm:sp}, Proposition~\ref{prop:unique} determines the set of strategy-proof, envy-free and Pareto efficient rules on every domain $\mathcal{D}'\subseteq\mathcal{D}$.

A preference $\pref\in\mathcal{D}$ is \emph{location-independent} if $(x;q)\sim(x';q)$ for all $(x;q),(x';q)\in T$. Let $\mathcal{D}^0\subseteq\mathcal{D}$ denote the set of location-independent preferences; the preference represented by $u$ in Section~2 belongs to $\mathcal{D}^0$.

\begin{corollary}\label{cor:location}
Every ${\pref}\in\mathcal{D}^0$ has balanced cut $\tfrac12$. Hence, for every $\mathcal{D}'\subseteq\mathcal{D}^0$, the rule $f(\pref)=(\L(\tfrac12),\R(\tfrac12))$ is strategy-proof, envy-free and Pareto efficient on $\mathcal{D}'^2$.
\end{corollary}

\begin{proof}
For ${\pref}\in\mathcal{D}^0$, $\L(\tfrac12)=(0;\tfrac12)\sim(\tfrac12;\tfrac12)=\R(\tfrac12)$, so $\qf(\pref)=\tfrac12$ by Proposition~\ref{prop:cut}. Hence $Q(\mathcal{D}')=\{\tfrac12\}$, and the claim follows from Proposition~\ref{prop:unique}(i).
\end{proof}

Consequently, in the necessity part of Theorem~\ref{thm:sp}, at least two of the three preferences are not location-independent.

\section{Preferences single-peaked in quantity}

A preference in $\mathcal{S}\setminus\mathcal{D}$ violates A2: some bundle is not strictly preferred to one of its proper subintervals. The following result shows that such preferences lie outside the additive model of cake cutting.

\begin{proposition}\label{prop:nonadditive}
Let ${\pref}\in\mathcal{S}$. If there is a finite signed measure $\mu$ on $[0,1]$ such that $(x;q)\pref(x';q')$ if and only if $\mu([x,x+q])\ge\mu([x',x'+q'])$ for all $(x;q),(x';q')\in T$, then ${\pref}\in\mathcal{D}$.
\end{proposition}

Theorem~\ref{thm:sp} implies the following.

\begin{corollary}\label{cor:ssp}
If $\mathcal{D}'\subseteq\mathcal{S}$ contains three monotonic preferences with distinct balanced cuts, there is no strategy-proof, envy-free and Pareto efficient rule on $\mathcal{D}'^2$. In particular, there is none on $\mathcal{S}^2$.
\end{corollary}

\begin{proof}
Let $\mathcal{D}''\subseteq\mathcal{D}'\cap\mathcal{D}$ consist of the three monotonic preferences. If $f$ were SP, EF and PE on $\mathcal{D}'^2$, its restriction to $\mathcal{D}''^2$ would be EF and PE (both are conditions at each profile, and $\mathcal{A}$ is the same) and SP (fewer misreports are available), contradicting Theorem~\ref{thm:sp}.
\end{proof}

The next result identifies domains of preferences single-peaked in quantity on which the three properties are compatible. For $\pref\in\mathcal{S}$, let $\qt(\pref)$ and $\xt(\pref)$ be as in A3. For $c\in(0,1)$, let
\[
\mathcal{S}_c=\{{\pref}\in\mathcal{S}:\ \qt(\pref)\le c\le\xt(\pref)\}
\]
denote the set of preferences whose most preferred left piece is contained in $[0,c]$ and whose most preferred right piece is contained in $[c,1]$. The \emph{top rule} assigns to each agent a most preferred bundle:
\[
f^\tau(\pref)=\big((0;\qt(\pref_1)),\ (\xt(\pref_2);1-\xt(\pref_2))\big).
\]

\begin{proposition}\label{prop:top}
For every $c\in(0,1)$, the top rule $f^\tau:\mathcal{S}_c^2\to\mathcal{A}$ is strategy-proof, envy-free and Pareto efficient.
\end{proposition}

The domain $\mathcal{S}_c$ contains preferences that depend on location (Example~\ref{ex:top} in the Appendix). Since $\qt(\pref)=1>c$ for every ${\pref}\in\mathcal{D}$, it contains no monotonic preference and is therefore not covered by Corollary~\ref{cor:ssp}.

\section{Concluding remarks}

For any number of agents with monotonic preferences, envy-free and Pareto efficient allocations into intervals exist: envy-free partitions of $[0,1]$ into intervals exist \citep{stromquist1980cut, su1999rental}, and every such partition is Pareto efficient among all allocations into intervals (Lemma~\ref{lem:nagents}). Three questions remain open: whether Theorem~\ref{thm:sp} and Proposition~\ref{prop:unique} extend to more than two agents; which subdomains of $\mathcal{S}$ admit a strategy-proof, envy-free and Pareto efficient rule; and the structure of envy-free and Pareto efficient allocations on $\mathcal{S}$.

\appendix
\section*{Appendix: Proofs}
\setcounter{section}{1}
\renewcommand{\thesection}{A}

Throughout, ${\pref}\in\mathcal{D}$ unless stated otherwise.

\begin{lemma}\label{lem:mono}
(a) If $0\le s<t\le1$, then $\L(t)\succ\L(s)$ and $\R(s)\succ\R(t)$. (b) Every bundle of positive quantity is strictly preferred to every bundle of zero quantity.
\end{lemma}

\begin{proof}
(a) $[0,s]\subsetneq[0,t]$ and $[t,1]\subsetneq[s,1]$; apply A2. (b) Let $(x;q)\in T$ with $q>0$ and $s\in[0,1]$. Then $[x,x]\subsetneq[x,x+q]$, so $(x;q)\succ(x;0)\sim(s;0)$ by A2 and A1.
\end{proof}

\begin{proof}[Proof of Proposition~\ref{prop:cut}]
Since $T\subset\mathbb{R}^2$ has a countable base and the contour sets of $\pref$ are closed, $\pref$ has a continuous utility representation $u$ \citep{debreu1954representation}. Let $g(t)=u(\L(t))-u(\R(t))$ for $t\in[0,1]$. Then $g$ is continuous, and strictly increasing by Lemma~\ref{lem:mono}(a). Moreover $g(0)=u(0;0)-u(0;1)<0$ and $g(1)=u(0;1)-u(1;0)>0$ by Lemma~\ref{lem:mono}(b). By the intermediate value theorem there is a unique $\qf\in(0,1)$ with $g(\qf)=0$, and $g(t)<0$ for $t<\qf$ and $g(t)>0$ for $t>\qf$. The conclusions do not depend on the choice of $u$.
\end{proof}

\begin{lemma}\label{lem:whole}
If $y\in\mathcal{A}$ and $y_1\cup y_2\neq[0,1]$, there is $z\in\mathcal{A}$ with $z_1\cup z_2=[0,1]$, $z_k\supseteq y_k$ for both $k$, and $z_k\supsetneq y_k$ for some $k$. Consequently, at every profile in $\mathcal{D}^2$, a Pareto efficient allocation allocates the whole good, and an allocation that allocates the whole good is of one of the forms (W1) one agent receives $(0;1)$ and the other a bundle of zero quantity; (W2) $(\L(\alpha),\R(\alpha))$ with $\alpha\in(0,1)$; (W3) $(\R(\alpha),\L(\alpha))$ with $\alpha\in(0,1)$.
\end{lemma}

\begin{proof}
Write $y_k=[a_k,b_k]$. If some bundle, say $y_i$, has zero quantity, let the other agent $j$ receive $z_j=[0,1]\supsetneq y_j$ and let $z_i=y_i$. Otherwise both bundles have positive quantity; then $a_1\neq a_2$, since equal left ends would make the interiors overlap. Let $i$ be the agent with the smaller left end and $j$ the other. Since $a_i<a_j$, feasibility gives $b_i\le a_j$. Let $z_i=[0,a_j]$ and $z_j=[a_j,1]$. Then $z\in\mathcal{A}$, $y_i\subseteq z_i$ and $y_j\subseteq z_j$, and both inclusions are equalities only if $a_i=0$, $b_i=a_j$ and $b_j=1$, that is, only if $y_1\cup y_2=[0,1]$, which is excluded. In each case A2 gives $z_k\pref_k y_k$ for both $k$, with strict preference for some $k$, so $y$ is not Pareto efficient. Finally, let $y_1\cup y_2=[0,1]$. If $y_i$ has zero quantity, then $y_j=[0,1]$: this is (W1). If both quantities are positive, then with $i,j$ as above, $a_i<a_j$ and $b_i\le a_j$ force $a_i=0$, $b_i=a_j=:\alpha\in(0,1)$ and $b_j=1$: this is (W2) if $i=1$ and (W3) if $i=2$.
\end{proof}

\begin{proof}[Proof of Proposition~\ref{prop:fair}]
\emph{Necessity.} Let $y$ be PE and EF. By Lemma~\ref{lem:whole}, $y$ is of form (W1), (W2) or (W3). Under (W1), the agent with zero quantity envies the other by Lemma~\ref{lem:mono}(b). Under (W2), envy-freeness requires $\L(\alpha)\pref_1\R(\alpha)$ and $\R(\alpha)\pref_2\L(\alpha)$, that is, $\qf_1\le\alpha\le\qf_2$ by Proposition~\ref{prop:cut}. Under (W3) it requires $\qf_2\le\alpha\le\qf_1$. Reading off the cases $\qf_1=\qf_2$, $\qf_1<\qf_2$ and $\qf_1>\qf_2$ gives (i)--(iii). The same argument shows that every envy-free allocation that allocates the whole good is one of (i)--(iii).

\emph{Sufficiency.} Envy-freeness of (i)--(iii) follows from Proposition~\ref{prop:cut}. For Pareto efficiency, consider $w=(\L(\alpha),\R(\alpha))$ with $\qf_1\le\alpha\le\qf_2$; the case $(\R(\alpha),\L(\alpha))$ with $\qf_2\le\alpha\le\qf_1$ is symmetric. Suppose $z\in\mathcal{A}$ Pareto dominates $w$. By Lemma~\ref{lem:whole} and transitivity, we may assume that $z$ allocates the whole good. If $z$ is of form (W1), the agent with zero quantity is strictly worse off than under $w$, where both quantities are positive (Lemma~\ref{lem:mono}(b)). If $z=(\L(\beta),\R(\beta))$, then $\L(\beta)\pref_1\L(\alpha)$ requires $\beta\ge\alpha$ and $\R(\beta)\pref_2\R(\alpha)$ requires $\beta\le\alpha$ (Lemma~\ref{lem:mono}(a)), so $z=w$. If $z=(\R(\beta),\L(\beta))$, then $\R(\beta)\pref_1\L(\alpha)\pref_1\R(\alpha)$ gives $\beta\le\alpha$, and $\L(\beta)\pref_2\R(\alpha)\pref_2\L(\alpha)$ gives $\beta\ge\alpha$; hence $\beta=\alpha$. A strict gain for agent~1 would then require $\R(\alpha)\succ_1\L(\alpha)$, that is, $\alpha<\qf_1$, and a strict gain for agent~2 would require $\L(\alpha)\succ_2\R(\alpha)$, that is, $\alpha>\qf_2$; both are false. Hence no $z$ Pareto dominates $w$.
\end{proof}

\begin{lemma}\label{lem:nagents}
Let $n\ge2$ and let $\pref_1,\dots,\pref_n\in\mathcal{D}$. Call $(Y_1,\dots,Y_n)$ an \emph{$n$-allocation} if each $Y_k\in T$ and the interiors of the $Y_k$ are pairwise disjoint. If an $n$-allocation $P$ satisfies $P_1\cup\dots\cup P_n=[0,1]$ and $P_k\pref_k P_h$ for all $k,h$, then there is no $n$-allocation $Q$ with $Q_k\pref_k P_k$ for all $k$ and $Q_j\succ_j P_j$ for some $j$.
\end{lemma}

\begin{proof}
Each $P_k$ has positive quantity: otherwise, since the $P_h$ cover $[0,1]$, some $P_h$ has positive quantity and $P_h\succ_k P_k$ by Lemma~\ref{lem:mono}(b), contradicting $P_k\pref_k P_h$. Hence $P$ consists of $n$ intervals $P^{(m)}=[z_{m-1},z_m]$, $m=1,\dots,n$, with $0=z_0<z_1<\dots<z_n=1$. Let $Q$ be an $n$-allocation with $Q_k\pref_k P_k$ for all $k$. Each $Q_k$ has positive quantity, since otherwise $P_k\succ_k Q_k$ by Lemma~\ref{lem:mono}(b). Order the bundles of $Q$ from left to right as $Q^{(m)}=[a_m,b_m]$, so that $b_m\le a_{m+1}$, and set $b_0=0$.

We show by induction on $m$ that $Q^{(m)}=P^{(m)}$. Suppose $Q^{(m')}=P^{(m')}$ for all $m'<m$, so that $b_{m-1}=z_{m-1}$. Let $m^*\ge m$ be the smallest index with $b_{m^*}\le z_{m^*}$; it exists because $b_n\le1=z_n$. If $m^*>m$, then $b_{m^*-1}>z_{m^*-1}$, so $a_{m^*}>z_{m^*-1}$ and $Q^{(m^*)}\subsetneq P^{(m^*)}$. If $m^*=m$, then $a_m\ge b_{m-1}=z_{m-1}$ and $b_m\le z_m$, so $Q^{(m)}\subseteq P^{(m)}$. In either case $Q^{(m^*)}\subseteq P^{(m^*)}$; let agent $k$ receive $Q^{(m^*)}$ under $Q$ and agent $h$ receive $P^{(m^*)}$ under $P$. By A2, $P_h\pref_k Q_k$, and therefore
\[
P_h\pref_k Q_k\pref_k P_k\pref_k P_h.
\]
Hence $P_h\sim_k Q_k$, and A2 implies $Q_k=P_h$, that is, $Q^{(m^*)}=P^{(m^*)}$. This excludes the case $m^*>m$, so $m^*=m$ and $Q^{(m)}=P^{(m)}$.

Thus $Q$ assigns the bundles of $P$ to the agents, and for every $k$ there is $h$ with $Q_k=P_h$. Envy-freeness of $P$ gives $P_k\pref_k Q_k$, so $Q_k\sim_k P_k$ for all $k$, and $Q$ does not Pareto dominate $P$.
\end{proof}

\begin{proof}[Proof of Theorem~\ref{thm:sp}]
\emph{If.} This is shown in the text.

\emph{Only if.} Suppose $P_l,P_m,P_h\in\mathcal{D}'$ have balanced cuts $l<m<h$ and $f:\mathcal{D}'^2\to\mathcal{A}$ is SP, EF and PE. If $f(P_m,P_m)=(\R(m),\L(m))$, replace $f$ by the rule $\tilde f(\pref_1,\pref_2)=(f_2(\pref_2,\pref_1),f_1(\pref_2,\pref_1))$, obtained by exchanging the names of the agents, which is also SP, EF and PE. So by Proposition~\ref{prop:fair}(i) we may assume $f(P_m,P_m)=(\L(m),\R(m))$.

\emph{Step 1.} By Proposition~\ref{prop:fair}(ii), $f(P_l,P_m)=(\L(\alpha),\R(\alpha))$ with $\alpha\in[l,m]$. If agent~1 reports $P_m$, she receives $\L(m)$. SP requires $\L(\alpha)\pref_{P_l}\L(m)$, so $\alpha\ge m$ by Lemma~\ref{lem:mono}(a). Hence $f(P_l,P_m)=(\L(m),\R(m))$.

\emph{Step 2.} By Proposition~\ref{prop:fair}(i), $f(P_l,P_l)\in\{(\L(l),\R(l)),(\R(l),\L(l))\}$. At $(P_l,P_m)$, agent~2 receives $\R(m)$ and could report $P_l$. If $f_2(P_l,P_l)=\R(l)$, then $\R(l)\succ_{P_m}\R(m)$ since $l<m$, a profitable manipulation. Hence $f(P_l,P_l)=(\R(l),\L(l))$.

\emph{Step 3.} By Proposition~\ref{prop:fair}(iii), $f(P_h,P_l)=(\R(\beta),\L(\beta))$ with $\beta\in[l,h]$. If agent~1 reports $P_l$, she receives $\R(l)$ by Step~2. SP requires $\R(\beta)\pref_{P_h}\R(l)$, so $\beta\le l$. Hence $f_2(P_h,P_l)=\L(l)$.

\emph{Step 4.} By Proposition~\ref{prop:fair}(iii), $f_2(P_h,P_m)=\L(\gamma)$ with $\gamma\in[m,h]$. At $(P_h,P_l)$, agent~2 receives $\L(l)$ by Step~3; by reporting $P_m$ she receives $\L(\gamma)$ with $\gamma\ge m>l$, and $\L(\gamma)\succ_{P_l}\L(l)$ by Lemma~\ref{lem:mono}(a). This contradicts SP.
\end{proof}

\begin{proof}[Proof of Corollary~\ref{cor:sp}]
The preferences represented by $u$, $u'$ and $u''$ of Section~2 have balanced cuts $\tfrac12$, $\tfrac{3-\sqrt5}{2}$ and $\tfrac{\sqrt5-1}{2}$, which solve $t=1-t$, $t=(1-t)^2$ and $t^2=1-t$, respectively. For the additive preference with density $\tfrac1{2c}$ on $[0,c]$ and $\tfrac1{2(1-c)}$ on $[c,1]$, the measure of every non-degenerate interval is positive, so A1 and A2 hold, and $[0,c]$ and $[c,1]$ both have measure $\tfrac12$. The proof now follows from Theorem~\ref{thm:sp}.
\end{proof}

\begin{remark}[Independence of the axioms]\label{rem:indep}
Each pair of the three properties is satisfied by some rule on $\mathcal{D}^2$. The dictatorship $f(\pref)=((0;1),(1;0))$ is SP and PE but not EF. The rule that divides the good at $\qf_1$ and assigns $\L(\qf_1)$ to agent~1 if $\qf_1\le\qf_2$ and $\R(\qf_1)$ otherwise is EF and PE by Proposition~\ref{prop:fair}, and hence not SP by Corollary~\ref{cor:sp}. The rule $f(\pref)=((0;0),(0;0))$ is SP and EF but not PE.
\end{remark}

\begin{proof}[Proof of Proposition~\ref{prop:unique}]
(i) Every profile in $\mathcal{D}'^2$ satisfies $\qf_1=\qf_2=c$, so by Proposition~\ref{prop:fair}(i) a rule is envy-free and Pareto efficient if and only if $f(\pref)\in\{(\L(c),\R(c)),(\R(c),\L(c))\}$ for every $\pref$. Every such rule is strategy-proof: whatever an agent reports, she receives $\L(c)$ or $\R(c)$, and $\L(c)\sim_i\R(c)$ for both agents.

(ii) \emph{If.} For $f^1$ this is shown in the proof of Theorem~\ref{thm:sp}. The rule $f^2$ is obtained from $f^1$ by exchanging the roles of the agents, that is, $f^2(\pref_1,\pref_2)=(f^1_2(\pref_2,\pref_1),f^1_1(\pref_2,\pref_1))$; it therefore has the same properties.

\emph{Only if.} Let $f$ be strategy-proof, envy-free and Pareto efficient. For $a,b\in\{l,h\}$, let $\Pi_{ab}$ denote the set of profiles with $\qf_1=a$ and $\qf_2=b$; each $\Pi_{ab}$ is nonempty. Throughout, preferences with balanced cut $l$ are denoted $P,P'$ and preferences with balanced cut $h$ are denoted $H,H'$.

\emph{Step 1: there are $\alpha,\beta\in[l,h]$ with $f(\pref)=(\L(\alpha),\R(\alpha))$ for all ${\pref}\in\Pi_{lh}$ and $f(\pref)=(\R(\beta),\L(\beta))$ for all ${\pref}\in\Pi_{hl}$.} By Proposition~\ref{prop:fair}(ii), $f(P,H)=(\L(\alpha(P,H)),\R(\alpha(P,H)))$ with $\alpha(P,H)\in[l,h]$. Strategy-proofness for agent~1 with preference $P$ requires $\L(\alpha(P,H))\pref_P\L(\alpha(P',H))$, so $\alpha(P,H)\ge\alpha(P',H)$ by Lemma~\ref{lem:mono}(a); exchanging $P$ and $P'$ gives equality. Strategy-proofness for agent~2 with preference $H$ requires $\R(\alpha(P,H))\pref_H\R(\alpha(P,H'))$, so $\alpha(P,H)\le\alpha(P,H')$; exchanging $H$ and $H'$ gives equality. Hence $\alpha(P,H)=\alpha(P',H)=\alpha(P',H')$ for all $P,P',H,H'$. The argument for $\Pi_{hl}$ is symmetric, using Proposition~\ref{prop:fair}(iii).

\emph{Step 2: $\alpha=l$ or $\beta=l$.} Let $(P,P')\in\Pi_{ll}$; by Proposition~\ref{prop:fair}(i), $f(P,P')\in\{(\L(l),\R(l)),(\R(l),\L(l))\}$. If $\beta>l$, then $f_1(P,P')=\L(l)$: otherwise agent~1 with preference $H$ at the profile $(H,P')\in\Pi_{hl}$ receives $\R(\beta)$, obtains $\R(l)$ by reporting $P$, and $\R(l)\succ_H\R(\beta)$ by Lemma~\ref{lem:mono}(a). If $\alpha>l$, then $f_2(P,P')=\L(l)$: otherwise agent~2 with preference $H$ at the profile $(P,H)\in\Pi_{lh}$ receives $\R(\alpha)$, obtains $\R(l)$ by reporting $P'$, and $\R(l)\succ_H\R(\alpha)$. Since $f_1(P,P')$ and $f_2(P,P')$ cannot both equal $\L(l)$, $\alpha>l$ and $\beta>l$ cannot both hold.

\emph{Step 3: $\alpha=h$ or $\beta=h$.} Let $(H,H')\in\Pi_{hh}$; by Proposition~\ref{prop:fair}(i), $f(H,H')\in\{(\L(h),\R(h)),(\R(h),\L(h))\}$. If $\alpha<h$, then $f_1(H,H')=\R(h)$: otherwise agent~1 with preference $P$ at $(P,H')\in\Pi_{lh}$ receives $\L(\alpha)$, obtains $\L(h)$ by reporting $H$, and $\L(h)\succ_P\L(\alpha)$. If $\beta<h$, then $f_2(H,H')=\R(h)$: otherwise agent~2 with preference $P$ at $(H,P)\in\Pi_{hl}$ receives $\L(\beta)$, obtains $\L(h)$ by reporting $H'$, and $\L(h)\succ_P\L(\beta)$. Since $f_1(H,H')$ and $f_2(H,H')$ cannot both equal $\R(h)$, $\alpha<h$ and $\beta<h$ cannot both hold.

\emph{Step 4: conclusion.} Since $l<h$, Steps~2 and~3 imply $(\alpha,\beta)\in\{(l,h),(h,l)\}$. If $(\alpha,\beta)=(l,h)$, then $\beta>l$ gives $f(\pref)=(\L(l),\R(l))$ on $\Pi_{ll}$ and $\alpha<h$ gives $f(\pref)=(\R(h),\L(h))$ on $\Pi_{hh}$ by Steps~2 and~3; together with Step~1, $f=f^1$. If $(\alpha,\beta)=(h,l)$, then $\alpha>l$ gives $f(\pref)=(\R(l),\L(l))$ on $\Pi_{ll}$ and $\beta<h$ gives $f(\pref)=(\L(h),\R(h))$ on $\Pi_{hh}$; together with Step~1, $f=f^2$.
\end{proof}

\begin{proof}[Proof of Proposition~\ref{prop:nonadditive}]
Suppose such a $\mu$ exists. By A1, $\mu(\{x\})=\kappa$ for every $x\in[0,1]$; since a finite signed measure has at most countably many atoms, $\kappa=0$.

Suppose first that $\qt=1$. By A3(i), $\L(t)\succ\L(s)$ whenever $0\le s<t\le1$, so $\mu([s,t])=\mu([0,t])-\mu([0,s])>0$. If $[x',x'+q']\subsetneq[x,x+q]$, the difference $[x,x+q]\setminus[x',x'+q']$ contains a nondegenerate interval, so $\mu([x,x+q])>\mu([x',x'+q'])$ and A2 holds. Hence ${\pref}\in\mathcal{D}$.

Suppose now that $\qt<1$; we derive a contradiction. Write $a=\qt<1$, $b=\xt<1$ and $M=\mu([0,a])=\mu([b,1])$, the value of a most preferred bundle. By A2$'$, $M>\mu(\{0\})=0$. By A3(i), $\L(a)\succ\L(1)$, so $\mu([0,1])<M$.

If $b\le a$, then $\mu([0,1])=\mu([0,b])+M$, and $\mu([0,b])\ge0$ (it equals $0$ if $b=0$ and is positive by A2$'$ if $b>0$). Hence $\mu([0,1])\ge M$, a contradiction.

If $b>a$, then $\mu([0,b])=M+\mu([a,b])$, and A3(i) gives $\L(a)\succ\L(b)$, so $\mu([a,b])<0$; and $\mu([0,1])=M+\mu([a,b])+M$. Since $a<b<1$, A3(i) gives $\L(b)\succ\L(1)$, that is, $M+\mu([a,b])>2M+\mu([a,b])$, so $M<0$, a contradiction.
\end{proof}

\begin{proof}[Proof of Proposition~\ref{prop:top}]
Let ${\pref}\in\mathcal{S}_c^2$. Agent~1 receives $[0,\qt(\pref_1)]$ and agent~2 receives $[\xt(\pref_2),1]$, and $\qt(\pref_1)\le c\le\xt(\pref_2)$, so the interiors are disjoint. Each agent receives a bundle that is most preferred in $T$ under her reported preference. Under truthful reports, each agent's bundle is therefore weakly preferred to every bundle, which gives envy-freeness and Pareto efficiency, and no misreport can give an agent a bundle she strictly prefers, which gives strategy-proofness.
\end{proof}

\begin{example}\label{ex:top}
For $p\in(0,1)$, $u_p(x;q)=q\,e^{-q/p}$ represents a preference in $\mathcal{S}$ with $\qt=p$ and $\xt=1-p$; the preference is independent of location. It belongs to $\mathcal{S}_c$ if and only if $p\le\min\{c,1-c\}$. Membership in $\mathcal{S}_c$ restricts only $\qt$ and $\xt$, so $\mathcal{S}_c$ also contains preferences that depend on location. For example, $u(x;q)=r\,e^{1-r}$ with $r=q/(p(1+x))$ represents a preference with $\qt=p$ and $\xt=\tfrac{1-p}{1+p}$, which belongs to $\mathcal{S}_c$ if and only if $p\le c\le\tfrac{1-p}{1+p}$; such $c$ exists if and only if $p\le\sqrt2-1$.
\end{example}

\paragraph{Verification of Example~\ref{ex:top}.}
For fixed $x$, $q\mapsto q\,e^{-q/p}$ is strictly increasing on $[0,p]$ and strictly decreasing on $[p,1]$, and its value is $0$ at $q=0$ and positive for $q>0$. Hence A1 and A2$'$ hold, the top indifference curve is $\{(x;p):x\in[0,1-p]\}$, so $\qt=p$ and $\xt=1-p$, and A3 holds because along $x=0$ and along $x+q=1$ the utility is single-peaked in the quantity, with peak $p$. Hence $u_p\in\mathcal{S}_c$ if and only if $p\le c\le1-p$.

For $u=\varphi(r)$ with $\varphi(r)=r\,e^{1-r}$ and $r=q/(p(1+x))$: $u$ is continuous on $T$, and $\varphi$ is strictly increasing on $[0,1]$ and strictly decreasing on $[1,\infty)$, with $\varphi(0)=0$ and maximum $\varphi(1)=1$. Hence $u(x;0)=0$ (A1) and $u(0;q)=\varphi(q/p)>0$ for $q>0$ (A2$'$). The most preferred bundles are those with $r=1$, that is, $q=p(1+x)$ with $x+p(1+x)\le1$, that is, $x\le\tfrac{1-p}{1+p}$; so $\qt=p$, $\xt=\tfrac{1-p}{1+p}$. Along $x=0$, $u=\varphi(q/p)$ is single-peaked at $q=p$, which gives A3(i). Along $x+q=1$, $r(x)=\tfrac{1-x}{p(1+x)}$ is strictly decreasing with $r(\xt)=1$; so for $\xt\le x'<x''$ we have $r(x'')<r(x')\le1$, and for $x''<x'\le\xt$ we have $r(x'')>r(x')\ge1$; in both cases $\varphi(r(x'))>\varphi(r(x''))$, which gives A3(ii). Hence $u\in\mathcal{S}_c$ if and only if $p\le c\le\tfrac{1-p}{1+p}$, and such $c$ exists if and only if $p\le\tfrac{1-p}{1+p}$, that is, $p^2+2p-1\le0$, that is, $p\le\sqrt2-1$.

\end{document}